\documentclass[a4paper,fleqn]{article}
\usepackage{a4wide}
\usepackage{amsmath,amssymb,amsthm,graphics,psfrag}
\usepackage{thmtools}
\usepackage{enumerate}
\usepackage{hyperref}
\usepackage{bm}
\usepackage{authblk}
\usepackage{vaucanson-g}
\usepackage[T1]{fontenc}
\usepackage[utf8]{inputenc}

\declaretheorem[numberwithin=section,refname={theorem,theorems},Refname={Theorem,Theorems}]{theorem}
\declaretheorem[sibling=theorem,style=definition]{definition}
\declaretheorem[sibling=theorem,name=Lemma]{lemma}
\declaretheorem[sibling=theorem,name=Proposition]{proposition}
\declaretheorem[sibling=theorem,name=Corollary]{corollary}
\declaretheorem[style=remark,name=Example]{Example}
\declaretheorem[style=remark,name=Remark]{Remark}

\newcommand\blfootnote[1]{%
  \begingroup
  \renewcommand\thefootnote{}\footnote{#1}%
  \addtocounter{footnote}{-1}%
  \endgroup
}

\DeclareMathOperator{\rep}{rep}
\DeclareMathOperator{\val}{val}

\DeclareMathOperator{\Fact}{Fact}
\providecommand{\abs}[1]{\lvert#1\rvert}

\providecommand{\floor}[1]{\lfloor#1\rfloor}

\providecommand{\ceil}[1]{\lceil#1\rceil}

\newcommand{\infw}[1]{%
  \ifcat\noexpand#1\relax\bm{#1}% check if the argument is a control sequence, i.e., probably greek letter
  \else\mathbf{#1}\fi}          % should be a latin letter
\newcommand{\A}{\mathcal{A}}
\newcommand{\B}{\mathcal{B}}
\newcommand{\N}{\mathbb{N}}
\newcommand{\Z}{\mathbb{Z}}
\newcommand{\Q}{\mathbb{Q}}

\begin{document}
  \title{Automatic sequences based on Parry or Bertrand numeration systems}
  \author[1]{Adeline Massuir}
  \author[2,3,4]{Jarkko Peltomäki}
  \author[1]{Michel Rigo}
  \affil[1]{University of Liège, Department of Mathematics, Liège, Belgium}
  \affil[2]{The Turku Collegium for Science and Medicine TCSM, University of Turku, Turku, Finland}
  \affil[3]{Turku Centre for Computer Science TUCS, Turku, Finland}
  \affil[4]{University of Turku, Department of Mathematics and Statistics, Turku, Finland}
  \date{}
  \maketitle
  \blfootnote{E-mail addresses: \href{mailto:a.massuir@uliege.be}{a.massuir@uliege.be} (A. Massuir), \href{mailto:r@turambar.org}{r@turambar.org} (J. Peltomäki), \href{mailto:m.rigo@uliege.be}{m.rigo@uliege.be} (M. Rigo).} 
  \begin{abstract}
    We study the factor complexity and closure properties of automatic sequences based on Parry or Bertrand numeration
    systems. These automatic sequences can be viewed as generalizations of the more typical $k$-automatic sequences and
    Pisot-automatic sequences. We show that, like $k$-automatic sequences, Parry-automatic sequences have sublinear
    factor complexity while there exist Bertrand-automatic sequences with superlinear factor complexity. We prove that
    the set of Parry-automatic sequences with respect to a fixed Parry numeration system is not closed under taking
    images by uniform substitutions or periodic deletion of letters. These closure properties hold for $k$-automatic
    sequences and Pisot-automatic sequences, so our result shows that these properties are lost when generalizing to
    Parry numeration systems and beyond. Moreover, we show that a multidimensional sequence is $U$-automatic with
    respect to a positional numeration system $U$ with regular language of numeration if and only if its $U$-kernel is
    finite.
  \end{abstract}

\section{Introduction}
Roughly speaking, an automatic sequence is an infinite word over a finite alphabet such that its $n$th symbol is
obtained as the output given by a deterministic finite automaton fed with the representation of $n$ in a convenient
numeration system. Precise definitions are given in \autoref{ssec:as}.

If we consider the usual base-$k$ numeration systems, then we get the family of $k$-automatic sequences \cite{All}. These words are images under a coding of a fixed point of a substitution of constant length. On a larger scale, if one considers abstract numeration systems based on a regular language (see for instance \cite[Chap.~3]{cant} or \cite{Rigo}), then we get the family of morphic words. Between these two extremes, we have the automatic sequences based on Pisot, Parry, and Bertrand numeration systems (the definitions are given in \autoref{ssec:num}), and we have the following hierarchy:
\begin{eqnarray*}
  & \text{Integer base system}\subsetneq \text{Pisot systems} \subsetneq \text{Parry systems}\\ \subsetneq & \text{Bertrand systems with a regular numeration language} \subsetneq \text{Abstract numeration systems}.
\end{eqnarray*}
Abstract numeration systems are uniquely based on the genealogical ordering of the words belonging to a regular language. This is contrasting with the more restricted case, treated in this paper, of positional numeration systems based on an increasing sequence of integers: a digit occurring in $n$th position is multiplied by the $n$th element of the underlying sequence. 

The Pisot-automatic sequences behave in many respects like $k$-automatic sequences. Most importantly, in
both cases automatic sets have a characterization in terms of first-order logic. This characterization in the
$k$-automatic case is due to B\"uchi \cite{Buchi1960}, and it was generalized to the Pisot case by Bruy\`ere and Hansel
\cite{BH}; also see \cite[Chap.~3]{Rigo} and the references therein. Now by using the logical characterization, it is
particularly straightforward to show that both the class of $k$-automatic sequences and the class of Pisot-automatic
sequences enjoy many closure properties. For instance, both classes are closed under taking images by a uniform
substitution and under periodic deletion of letters; these are classical results of Cobham \cite{Cobham1972}. For more
closure properties, see \cite[Chap.~6.8]{All}.

In this paper, we study if some properties common to $k$-automatic sequences and Pisot-automatic sequences also hold for Parry-automatic sequences or more general automatic sequences. In a sense, we show that the generalization to Pisot numeration systems is the broadest possible generalization if the goal is to preserve the many good properties of $k$-automatic sequences.

It has been known before that a logical characterization no longer exists for Parry-automatic sequences \cite{Frougny1997}. We show that the closure properties mentioned above break when generalizing from Pisot to Parry and obtain as a corollary yet another proof showing that no logical characterization indeed exists for these sequences.

In combinatorics on words and in symbolic dynamics, the factor complexity of infinite words is often of interest. It was famously shown by Pansiot \cite{Pansiot1984} that the factor complexity of an infinite word generated by a substitution is in one of the following five classes: $\mathcal{O}(1)$, $\Theta(n)$, $\Theta(n \log n)$, $\Theta(n \log \log n)$, or $\Theta(n^2)$. Previously, it has been known that the factor complexity of a $k$-automatic sequence is sublinear (that is, it is in $\mathcal{O}(1)$ or $\Theta(n)$) \cite{Cobham1972}, \cite[Thm.~10.3.1]{All}. We extend this result and show that the factor complexity of any Parry-automatic sequence is sublinear. In contrast, we show by an explicit example that there exists a Bertrand-automatic sequence of superlinear complexity.

A well-known result concerning $k$-automatic sequences is their characterization in terms of the $k$-kernel originally due to Eilenberg \cite{Eilenberg}. This was generalized in \cite{RigoMaes} for all sequences associated with abstract numeration systems. The multidimensional version of this generalization \cite[Prop.~32]{RigoMaes} however needs an additional assumption that is not required in the $k$-automatic case. We show in this paper that this additional assumption is unnecessary also for positional numeration systems with a regular numeration language. 

%We recall in Section~\ref{ssec:num} how to canonically associate a numeration system with a real number such as a Pisot number or a Parry number. Pisot systems have many nice properties, e.g., characterizations of automatic sets in terms of first-order logic \cite{BH}, also see \cite{Rigo}. One interesting property of every Parry numeration system is that $w$ is a valid representation of an integer if and only if, for all $n\ge 0$, $w0^n$ also is. Such a property is called the Bertrand property \cite{Be}. The set of Parry numeration systems is a strict subset of the set of Bertrand numeration systems (see Remark~\ref{lem:strict_subset}).

This paper is organized as follows. In \autoref{sec:basics}, we recall needed results and notation on numeration systems and automatic sequences. Then in \autoref{sec:fac} we study the factor complexity of Parry-automatic sequences, and in \autoref{sec:closure}, we show that the closure properties of Pisot-automatic sequences do not hold for Parry-automatic sequences. The paper is concluded by \autoref{sec:multi}, where the relationship of $U$-automaticity and the finiteness of the $U$-kernel is studied in the multidimensional setting.

%Relying on Pansiot's theorem \cite{Pansiot1984}, we prove in Section~\ref{sec:fac} that automatic sequences built on a Parry system have a factor complexity in $\mathcal{O}(n)$. Then we exhibit an automatic word with quadratic factor complexity built on a system with a regular language of numeration having the Bertrand property, but which is not a Parry system. In Section~\ref{sec:closure}, we discuss a bit the closure properties of Parry-automatic sequences. 

\section{Basics}\label{sec:basics}
\subsection{Background on Numeration Systems}\label{ssec:num}

For general references on numeration systems and words, we refer the reader to \cite{cant,Lot,Rigo}. Let us first consider the representation of integers. A \emph{positional numeration system}, or simply, a \emph{numeration system}, is an increasing sequence $U=(U_n)_{n\ge 0}$ of integers such that $U_0=1$ and $C_U:=\sup_{n\ge 0}\ceil{U_{n+1}/U_n}<+\infty$. We let $A_U$ be the integer alphabet $\{0,\ldots,C_U-1\}$. The \emph{greedy representation} of the positive integer $n$ is the word $\rep_U(n)=w_{\ell-1}\cdots w_0$ over $A_U$ satisfying
\begin{equation*}
  \sum_{i=0}^{\ell-1} w_i\, U_i=n, \ w_{\ell-1}\neq 0,\ \text{ and }\ \forall j\in\{1,\ldots,\ell\},\quad \sum_{i=0}^{j-1}w_i\, U_i<U_j.
\end{equation*}
We set $\rep_U(0)$ to be the empty word $\varepsilon$. 
The language $\rep_U(\N)$ is called the \emph{numeration language}. A set
$X$ of integers is \emph{$U$-recognizable} if $\rep_U(X)$ is regular,
i.e., accepted by a finite automaton.  The \emph{numerical value} $\val_U\colon\Z^*\to\N$ maps a
word $d_{\ell-1}\cdots d_0$ over any alphabet of integers to the number $\sum_{i=0}^{\ell-1} d_iU_i$.

Recall that the \emph{genealogical ordering} orders words from a language first by length and then by the lexicographic
ordering (induced, in this paper, typically by the natural order of the digits).

\begin{definition}\label{def:bertrand}
  A numeration system $U$ is a \emph{Bertrand numeration system} if, for all $w \in A_U^*$, $w \in \rep_U(\N)$ if and
  only if $w0 \in \rep_U(\N)$.
\end{definition}

\begin{Example}\label{exa:1}
  The usual base-$k$ numeration system $(k^n)_{n\ge 0}$ is a Bertrand numeration system.
  Taking $F_0=1$, $F_1=2$, and $F_{n+2}=F_{n+1}+F_n$ for $n\ge 0$ gives the Fibonacci numeration system $F=(F_n)_{n\ge 0}$, which is a Bertrand numeration system: $\rep_F(\N)=1\{0,01\}^*\cup\{\varepsilon\}$. If we slightly modify the Fibonacci system by taking the initial conditions $U_0=1$,
  $U_1=3$, we get a numeration system $(U_n)_{n\ge 0}=(1,3,4,7,11,18,29,47,\ldots)$,
which is no longer a Bertrand system. Indeed, $2$ is the greedy representation of an integer but $20$ is not because $\rep_U(\val_U(20))=102$.
\end{Example}

\begin{Example}\label{exa:nonParry}
  The numeration system $B$ given by the recurrence $B_n = 3 B_{n-1} + 1$ for all $n \ge 1$ and $B_0 = 1$ is such that
  $0^*\rep_B(\N)=\{0,1,2\}^* (\{\varepsilon\} \cup 3\, 0^*)$; see \cite[p.~131]{Hollander1998}. The automaton accepting
  the language $0^* \rep_B(\N)$ is depicted in \autoref{fig:aaab}. By its simple form, it is obvious that it is a
  Bertrand numeration system. Notice that the sequence $(B_n)_{n\ge 0}$ also satisfies the homogeneous linear recurrence
  $B_n = 4 B_{n-1} - 3 B_{n-2}$.
\end{Example}

\begin{figure}
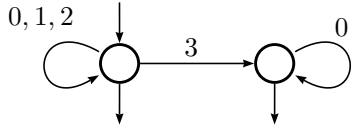

        \centering
        \VCDraw{%
        \begin{VCPicture}{(-2,-2)(2,2)}
% states
 \StateVar[]{(-1.7,0)}{s}
 \StateVar[]{(1.7,0)}{t} 
% initial--final
\Initial[n]{s}
\Final[s]{s}
\Final[s]{t}
% transitions 
\LoopW{s}{0,1,2}
\LoopE{t}{0}
\EdgeL{s}{t}{3}
\end{VCPicture}
}
        \caption{The canonical automaton accepting $\{0,1,2\}^* (\{\varepsilon\} \cup 3\, 0^*)$.}
        \label{fig:aaab}
\end{figure}

There is a link between the representation of integers and the representation of real numbers. Let $\beta>1$ be a real number. The \emph{$\beta$-expansion} of a real number $x\in[0,1]$ is the sequence $d_\beta(x)=(x_i)_{i\ge 1}\in\N^\omega$ that satisfies
\[x=\sum_{i=1}^{+\infty}x_i\beta^{-i}\]
and which is the maximal element in $\N^\omega$ having
this property with respect to the lexicographic order over $\N$. Notice that $\beta$-expansions can be obtained by a greedy algorithm and they only contain letters (digits) over the alphabet $A_\beta=\{0,\ldots,\ceil{\beta}-1\}$. By $\Fact(D_\beta)$, we denote the set of finite factors occurring in the
$\beta$-expansions of the real numbers in $[0,1)$. 

    \begin{definition}
        If $d_\beta(1)=t_1\cdots
t_m0^\omega$, with $t_1,\ldots,t_m\in A_\beta$ and $t_m\neq
0$, then we say that $d_\beta(1)$ is \emph{finite} and we set
$d_\beta^*(1)=(t_1\cdots t_{m-1}(t_m-1))^\omega$. Otherwise, we set
$d_\beta^*(1)=d_\beta(1)$. An equivalent definition is to set $d_\beta^*(1)=\lim_{x\to 1^-}d_\beta(x)$. When $d_\beta^*(1)$ is (ultimately) periodic,
then $\beta$ is said to be a \emph{Parry number}.
    \end{definition}

\begin{definition}
    Let $\beta>1$ be a real number such that $d_\beta^*(1)=(t_i)_{i\ge 1}$. The numeration system $U_\beta=(U_n)_{n\ge 0}$ canonically associated with $\beta$ is defined by 
    \begin{equation}
        \label{eq:bertrand}
        U_n=t_1U_{n-1}+\cdots +t_nU_0+1, \quad \forall n\ge 0.
    \end{equation}
\end{definition}
As a consequence of Bertrand's theorem, see \cite{Be} or \cite[Chap.~2]{cant}, the numeration system $U_\beta$ associated with $\beta$ satisfies  
\begin{equation}
    \label{eq:berthm}
    0^*\rep_{U_\beta}(\N)=\Fact(D_\beta).
\end{equation}
Thus for all $\beta > 1$, the canonical numeration system $U_\beta$ associated with $\beta$ is a Bertrand numeration system because $w \in \Fact(D_\beta)$ if and only if $w0 \in \Fact(D_\beta)$ for all $w \in A_\beta^*$.

    \begin{Remark}
 If $\beta$ is a Parry number, then the canonical numeration system $U_\beta$ satisfies a linear recurrence equation with integer coefficients which can be obtained from~\eqref{eq:bertrand}. 
    \end{Remark}

    \begin{definition}
        A numeration system $U$ is a \emph{Parry numeration system} if there exists a Parry number $\beta$ such that $U=U_\beta$.
    \end{definition}

\begin{Example}
  Consider the Golden mean $\beta=(1+\sqrt{5})/2$. We have $d_\beta(1)=11$ and $d_\beta^*(1)=(10)^\omega$, so the Golden mean is a Parry number. It is straightforward to deduce from \eqref{eq:bertrand} that the associated Parry numeration system is the Fibonacci system of \autoref{exa:1} defined by the recurrence $F_{n+2}=F_{n+1}+F_n$ and the initial conditions $F_0=1$, $F_1=2$.
\end{Example}

  \begin{lemma}\label{lem:strict_subset}
    The set of Parry numeration systems is a strict subset of the set of Bertrand numeration systems.
  \end{lemma}
\begin{proof}
  We have already deduced from \eqref{eq:berthm} that Parry numeration systems are Bertrand numeration systems. Now consider the Bertrand system $B=(B_n)_{n\ge 0}$ of \autoref{exa:nonParry}. We will show that there is no $\beta>1$ such that $B=U_\beta$. Proceed by contradiction. Assume that there exists $\beta$ such that $B$ is the numeration system canonically associated with $\beta$. The greatest word of length $n$ for the lexicographical order in $0^*\rep_B(\N)$ is $30^{n-1}$. Consequently, we have $1=3/\beta$. The Parry numeration system $U_3$ is the classical base-$3$ system and $0^*\rep_{U_3}(\N)=\{0,1,2\}^*$, which differs from $0^*\rep_B(\N)$. This is a contradiction.
\end{proof}

\begin{theorem}[Parry \cite{parry}]\label{the:parry}
    A sequence $x=(x_i)_{i\ge 1}$ over $\N$ is
    the $\beta$-expansion of a real number in $[0,1)$ if and only if
    $(x_{n+i})_{i\ge 1}$ is lexicographically less than 
    $ d_\beta^*(1)$ for all $n \geq 0$.
\end{theorem}

As a consequence of this result, with any Parry number $\beta$ is canonically associated a deterministic finite
automaton $\mathcal{A}_\beta=(Q_\beta,q_{\beta,0},F_\beta,A_\beta,\delta_\beta)$
accepting the language $\Fact(D_\beta)$. Otherwise stated, 
the numeration language of a Parry numeration system is regular. 
This automaton $\mathcal{A}_\beta$ has a special form. Let $d_\beta^*(1)=t_1\cdots
t_i(t_{i+1}\cdots t_{i+p})^\omega$ where $i\ge0$ and $p\ge1$ are the
minimal preperiod and period respectively. The set of states of $\mathcal{A}_\beta$
is $Q_\beta=\{q_{\beta,0},\ldots, q_{\beta,i+p-1}\}$. All states are
final. For every $j\in\{1,\ldots,i+p\}$, we have $t_j$ edges
$q_{\beta,j-1}\to q_{\beta,0}$ labeled by $0,\ldots,t_j-1$ and, for
$j<i+p$, one edge $q_{\beta,j-1}\to q_{\beta,j}$ labeled by
$t_j$. There is also an edge $q_{\beta,i+p-1} \to q_{\beta,i}$ labeled
by $t_{i+p}$.  See, for instance, \cite{FS,Rigo}. Note that in
\cite[Thm.~7.2.13]{Lot}, $\mathcal{A}_\beta$ is shown to be the trim
minimal automaton of $\Fact(D_\beta)$.

\begin{Remark}\label{rem:parry_prim}
    Let $\beta$ be a Parry number. The automaton $\mathcal{A}_\beta$ is well-known to be primitive (see, e.g., \cite[Chap.~3]{LindMarcus1995}). Indeed, the periodic part of $d_\beta^*(1)$ contains at least a nonzero digit. Consequently, there is a path from every state of $\mathcal{A}_\beta$ to the initial state $q_{\beta,0}$. Moreover, there is a loop on $q_{\beta,0}$ with label $0$. Hence $\mathcal{A}_\beta$ is irreducible (i.e., strongly connected) and aperiodic (the gcd of length of the cycles going through any state is $1$). The conclusion follows.
\end{Remark}

A \emph{Pisot number} is an algebraic integer $\beta>1$ whose conjugates have modulus less than $1$. A \emph{Salem number} is an is an algebraic integer $\beta>1$ whose conjugates have modulus less than or equal to $1$ and at least one has modulus equal to $1$. If $\beta$ is a Pisot number, then $U_\beta$ has many interesting properties \cite{BH,Fr,Rigo}: $\rep_{U_\beta}(\N)$ is regular, normalization $w\mapsto \rep_{U_\beta}(\val_{U_\beta}(w))$ (and thus addition) is computable by a finite automaton, and $U_\beta$-recognizable sets are characterized in terms of first order logic.

   \begin{definition}
     A numeration system $U$ is a \emph{Pisot numeration system} if there exists a Pisot number $\beta$ such that
     $U=U_\beta$.
   \end{definition}

   \begin{Remark}
     Pisot numbers are Parry numbers \cite{Schmidt1980}, so Pisot numeration systems belong to the set of Parry numeration
     systems. The numeration system~\eqref{eq:Un} studied in \autoref{sec:closure} is Parry but not Pisot, so this inclusion is
     strict. Further, the Fibonacci numeration system of \autoref{exa:1} is a Pisot numeration system because the
     Golden mean (the largest root of the polynomial $X^2 - X - 1$) is a Pisot number. As this numeration system is
     clearly not a base-$k$ numeration system, we see that base-$k$ numeration systems are strictly included in Pisot
     numeration systems.
   \end{Remark}

   \begin{Remark}
Let $\mathcal{B}$ be the set of Bertrand numeration systems. Let $\mathcal{R}$ be the set of numeration systems $U$
whose numeration language $\rep_U(\N)$ is regular. The three sets $\mathcal{B}\cap \mathcal{R}$, $\mathcal{B}\setminus
\mathcal{R}$ and $\mathcal{R}\setminus \mathcal{B}$ are nonempty. For instance, the modified Fibonacci system of
\autoref{exa:1} belongs to $\mathcal{R}\setminus \mathcal{B}$. All Parry numeration systems and the Bertrand system of
\autoref{exa:nonParry} belong to $\mathcal{B}\cap \mathcal{R}$.

        If $\beta$ is not a Parry number, for instance when $\beta$ is transcendental, then the numeration language $\rep_{U_\beta}(\N)$ is not regular even though $U_\beta$ is a Bertrand system. Hence $\mathcal{B}\setminus\mathcal{R}$ is nonempty. 

We will often make the assumption that we are dealing with positional numeration systems $U$ such that $\rep_U(\N)$ is
regular. This is particularly important when we will be dealing with finite $U$-kernels in \autoref{sec:multi}. This
assumption is in fact somewhat restrictive. In \cite{Shallit1994}, Shallit proves that if $\rep_U(\N)$ is regular, then
$(U_n)_{n\ge 0}$ must satisfy a homogeneous linear recurrence.
    \end{Remark}

\subsection{Automatic Sequences}\label{ssec:as}

\begin{definition}\label{def:automatic}
  Let $U$ be a numeration system. An infinite word $\infw{x}=(x_n)_{n\ge 0}$ over an alphabet $B$ is \emph{$U$-automatic}, i.e., it is an \emph{$U$-automatic sequence}, if there exists a complete DFAO (deterministic finite automaton with output) $(Q,q_0,A_U,\delta,\tau)$ with transition function $\delta\colon Q\times A_U\to Q$ and output function $\tau\colon Q\to B$ such that $\delta(q_0,0)=q_0$ and
  \begin{equation*}
    x_n=\tau(\delta(q_0,\rep_U(n))),\quad \forall n\ge 0.
  \end{equation*}
The infinite word $\infw{x}$ is \emph{$k$-automatic} (resp. \emph{Parry-automatic}, \emph{Bertrand-automatic}) if $U=(k^n)_{n\ge 0}$ for an integer $k \geq 2$ (resp. $U$ is a Parry numeration system, $U$ is a Bertrand numeration system).
\end{definition}

The next result is classical, see for instance \cite{Rigo2000}.

\begin{theorem}\label{the:Uautomatic}
  Let $U$ be a numeration system such that $\rep_U(\N)$ is regular.     An infinite word $\infw{x}=(x_n)_{n\ge 0}$ over $A$ is $U$-automatic if and only if, for all $a\in A$, the set $\{j\ge 0\mid x_j=a\}$ is $U$-recognizable.
\end{theorem}

Let $k\ge 2$ be an integer. The \emph{$k$-kernel} of an infinite word $\infw{x}=(x_n)_{n\ge 0}$ over $A$ is the set of
its subsequences of the form
\begin{equation*}
  \{(x_{k^e n+d})_{n\ge 0}\in A^\N \mid e\ge 0, 0\le d<k^e\}.
\end{equation*}
Observe that an element of the $k$-kernel is obtained by considering those indices whose base-$k$ expansions end with
$\rep_k(d)$ (possibly preceded by some zeroes to get a suffix of length $e$). With this in mind, we introduce the more
general $U$-kernel of an infinite word.

\begin{definition}\label{def:Ukernel}
  Let $U$ be a numeration system and $s \in A_U^*$ be a finite word. Define the ordered set of integers 
  \begin{equation}\label{eq:Is}
    \mathcal{I}_s:=\val_U(0^*\rep_U(\N)\cap A_U^*s)=\{i(s,0)<i(s,1)<\cdots\}.
  \end{equation}
  Depending on $s$, it is possible for this set to be finite or empty. The \emph{$U$-kernel} of an infinite word
  $\infw{x}=(x_n)_{n\ge 0}$ over $B$ is the set
  \begin{equation*}
    \ker_U(\infw{x}):=\{(x_{i(s,n)})_{n\ge 0} \mid s\in A_U^*\}.
  \end{equation*}
  With the above remark, this set can contain finite or even empty subsequences.
\end{definition}

The next two results have been obtained in the general framework of abstract numeration systems; see \cite[Prop.~7]{RigoMaes} and \cite[Prop.~9]{RigoMaes}.

\begin{proposition}\label{prp:kernel_1D}
  Let $U$ be a numeration system such that $\rep_U(\N)$ is regular. A word $\infw{x}$ is $U$-automatic if and only if its $U$-kernel is finite.    
\end{proposition}

\begin{proposition}\label{pro:reversal}
    Let $U$ be a numeration system. If an infinite word is $U$-automatic, then  it is reversal-$U$-automatic, i.e., its $n$th term is obtained by reading the reversal of $\rep_U(n)$ in a DFAO.
\end{proposition}

Notice that the proof of the latter result only relies on classical constructions on automata defined from the DFAO
generating the $U$-automatic sequence. The same construction applies in multidimensional setting, and we shall make use
of this in \autoref{sec:multi}.

  \section{Factor Complexity}\label{sec:fac}
  The \emph{factor complexity function} $p_\infw{x}(n)$ of an infinite word $\infw{x}$ counts the number of factors of
  length $n$ occurring in $\infw{x}$. For more on factor complexity, see \cite[Chap.~3]{cant}.

  Let us recall the following result of Cobham.

  \begin{proposition}
    \cite{Cobham1972}, \cite[Thm.~10.3.1]{All} The factor complexity function of a $k$-automatic sequence is sublinear.
  \end{proposition}

  Next we generalize this result.

  \begin{theorem}\label{thm:parry_sublinear}
    The factor complexity function of a Parry-automatic sequence is sublinear.
  \end{theorem}
  \begin{proof}
    Let $U$ be a Parry numeration system having canonical automaton $\A$, and let $\infw{x}$ be an $U$-automatic
    sequence generated by a DFAO $\B$. 
The product automaton $\A \times \B$ has $Q_{\A\times\B}:=Q_\A\times Q_\B$ as set of states, the initial state $q_0$ is the pair made of the initial states of $\A$ and $\B$, and the transition function is given by 
    \begin{equation*}
      \delta_{\A \times \B}((q,q'),i)=(\delta_\A(q,i),\delta_\B(q',i)).
    \end{equation*}
We consider the automaton $\A \times \B$ as a DFAO by setting that the output
    function $\tau$ maps a state $(q_\A, q_\B)$ of $\A \times \B$ to the output of the state $q_\B$ of $\B$. It is
    clear that $\infw{x}$ is generated by $\A \times \B$.

    Based on the automaton $\A \times \B$, we can build a substitution $\sigma$ and consider the output function $\tau$
    as a coding such that $\infw{x} = \tau(\sigma^\omega(q))$ for some $q\in Q_{\A\times\B}$. The construction is
    classical, see for instance \cite[Lemma~2.28]{Rigo}. The substitution $\sigma$ is defined as follows
    \begin{equation*}
      \sigma((q_\A, q_\B))=(\delta_\A(q_\A,0), \delta_\B(q_\B,0))(\delta_\A(q_\A,1), \delta_\B(q_\B,1)) \cdots (\delta_\A(q_\A,C_U-1), \delta_\B(q_\B,C_U-1)).
    \end{equation*}
    In the latter expression, since $\mathcal{A}$ is in general not complete, if $\delta_\A(q_\A,j)$ is undefined,
    then $(\delta_\A(q_\A,j), \delta_\B(q_\B,j))$ is replaced by $\varepsilon$. Notice that the substitution $\sigma$ is defined over the alphabet $Q_{\A\times\B}$.  Since $\A \times \B$ has a loop with
    label $(0,0)$ on its initial state $q_0$, iterating $\sigma$ on this state generates the sequence of states
    $\sigma^\omega(q_0)$ in $\A \times \B$ reached from the initial state by the words of $\rep_U(\N)$ in genealogical
    order.

Since every state of a canonical automaton of a Parry
    numeration system is final, the coding $\tau$ is non-erasing. Then by \cite[Lemma~4.6.6]{cant}, the factor
    complexity of $\infw{x}$ is at most the factor complexity of $\sigma^\omega(q_0)$, so it is sufficient to show that a
    fixed point of $\sigma$ has sublinear complexity. This is accomplished as follows. First we establish that there
    exists a number $\alpha$ such that $\abs{\sigma^n(q)} = \Theta(\alpha^n)$ for every state $q\in Q_{\A\times\B}$. In other words, we
    show that the substitution $\sigma$ is \emph{quasi-uniform}. It then follows from the famous theorem of Pansiot
    \cite{Pansiot1984} that the factor complexity of a fixed point of $\sigma$ is sublinear; also see
    \cite[Thm.~4.7.47]{cant}.

    Let us define a projection mapping $\varphi\colon Q_{\A \times \B} \to Q_\A$ by setting $\varphi((q_\A, q_\B)) = q_\A$ if
    $(q_\A, q_\B)$ is a state of $\A \times \B$. By the definition of the product automaton $\A \times \B$, we have
    $\varphi(\delta_{\A \times \B}((q_\A, q_\B), a)) = \delta_\A(\varphi((q_\A, q_\B)), a)$ for all letters $a$ and all
    states $q_\A$ and $q_\B$. 
% As there exists a path from the initial state to $q_\A$ in $\A$ and the automaton $\B$ is complete, there exists a state $(q_\A, q_\B)$ in $\A \times \B$ for some state $q_\B$ of $\B$. Thus the map $\varphi$ is surjective.

    Recall that given an automaton $\mathcal{C}$ with adjacency matrix $Adj(\mathcal{C})$, the entry $(Adj(\mathcal{C}))^n_{i,j}$ counts the number of distinct paths of length
    $n$ from state $i$ to state $j$; see \cite[Chap.~2]{LindMarcus1995}. Let $(q_\A, q_\B)$ be a state of
    $\A \times \B$ and consider all paths of length $n$ starting from this state. These paths can be identified with
    their edge labels. Given such a path with edge label $w$, we find by applying the projection mapping $\varphi$ a
    path in $\A$ with edge label $w$ starting at the state $q_\A$. Conversely, given a path of length $n$ in $\A$ with
    edge label $w$ starting at state $q_\A$, there exists a path with edge label $w$ in $\A \times \B$ starting at the
    state $(q_\A, q_\B)$ because the automaton $\B$ is complete (see \autoref{def:automatic}). Denoting the total number of paths of length $n$
    starting at a state $q$ of $\A \times \B$ by $K_q(n)$, we have thus argued that
    \begin{equation*}
       K_q(n) = \sum_{r \in Q_{\A \times \B}} (Adj(\A \times \B))^n_{q,r} = \sum_{s \in Q_{\A}} (Adj(\A))^n_{\varphi(q), s}.
    \end{equation*}
    The canonical automaton of a Parry numeration system is primitive (recall \autoref{rem:parry_prim}), we have for each $i$ and $j$ that $((Adj(\A))^n_{i,j} = \Theta(\alpha^n)$, where
    $\alpha$ is the Perron--Frobenius eigenvalue of $\A$. Thus $K_q(n) = \Theta(\alpha^n)$. By rephrasing the number
    $K_q(n)$ in terms of substitutions, we have $\abs{\sigma^n(q)} = K_q(n)$. Hence
    $\abs{\sigma^n(q)} = \Theta(\alpha^n)$, and we are done.
  \end{proof}

  Notice that in fact we showed in the proof of \autoref{thm:parry_sublinear} that for each Parry-automatic sequence
  $\infw{x}$ there exists a coding $\tau$ and a quasi-uniform substitution $\sigma$ such that
  $\infw{x} = \tau(\sigma^\omega(a))$ for a letter $a$. This should be contrasted with the fact that $k$-automatic
  sequences are codings of fixed points of \emph{uniform} substitutions.

From \autoref{lem:strict_subset}, there exist Bertrand numeration systems that are not Parry numeration systems. We show that \autoref{thm:parry_sublinear} does not generalize to Bertrand-automatic sequences.

  \begin{theorem}
    There exists a Bertrand-automatic sequence with superlinear factor complexity.
  \end{theorem}
  \begin{proof}
    Consider the numeration system given by the recurrence $B_n = 3 B_{n-1} + 1$ with
    $B_0 = 1$. In \autoref{exa:nonParry}, it was shown that this numeration system is a Bertrand
    numeration system.

    The substitution associated with the canonical automaton, depicted in \autoref{fig:aaab},
    is $\sigma\colon a \mapsto aaab$, $b \mapsto b$; see the proof of \autoref{thm:parry_sublinear} how this substitution is
    defined. Let $\infw{x}$ be the infinite fixed point of $\sigma$. Observe that $\infw{x}$ is
    a Bertrand-automatic sequence. It is easy to see that $ab^n a$ occurs in $\infw{x}$ for all $n \geq 0$. Thus $\infw{x}$
    is aperiodic and there exists infinitely many bounded factors occurring in $\infw{x}$ (a factor is $w$
    \emph{bounded} if the sequence $(\abs{\sigma^n(w)})_{n\ge 0}$ is bounded). It follows by Pansiot's theorem \cite{Pansiot1984}
    that the factor complexity of $\infw{x}$ is quadratic; see also \cite[Thm.~4.7.66]{cant}.
  \end{proof}

  We do not have examples of Bertrand-automatic sequences with factor complexities $\Theta(n \log n)$ or
  $\Theta(n \log \log n)$.

\section{Closure Properties}\label{sec:closure}

It is easy to see that the image of a $k$-automatic sequence $\infw{x}\in A^\omega$ under a substitution $\mu\colon A\to B^*$ of constant length $\ell$ is again a $k$-automatic sequence. Indeed, \autoref{the:Uautomatic} implies that for all $a\in A$ there exists a first-order formula $\varphi_a(n)$ in $\langle \N,+,V_k\rangle$ which holds if and only if $\infw{x}[n]=a$. Let us then define for each $b \in B$ a formula $\psi_b(n)$ that holds if and only if $\mu(\infw{x})[n]=b$. For each $n$ there exist unique $q$ and $r$ such that $0 \leq r < \ell$ and $n = \ell q + r$. For each $a \in A$, we can construct a formula $\sigma_a(r)$ that holds if and only if $\mu(a)$ contains the letter $b$ at position $r$ (indexing from $0$). Setting
\begin{equation*}
  \psi_b(n) = (\exists q)(\exists r < \ell)(n = \ell q + r \wedge \bigvee_{a \in A}(\varphi_a(q) \wedge \sigma_a(r)))
\end{equation*}
certainly has the desired effect. Notice that this is indeed a formula in $\langle \N, +, V_k \rangle$ since $\ell$ is constant. Therefore it follows from \autoref{the:Uautomatic} that $\mu(\infw{x})$ is $k$-automatic.

\begin{Example}
    Assume $A=\{a,b\}$, $B=\{c,d\}$, $\ell=3$ and set $\mu(a)=ccd$, $\mu(b)=dcd$. In this case, the formula $\psi_c(n)$ is given by
    \begin{equation*}
      (\exists q)(\exists r<3)(n=3q+r\wedge [(\varphi_a(q)\wedge(r=0\vee r=1))\vee (\varphi_b(q)\wedge r=1)]).
    \end{equation*}
\end{Example}
The same construction can be applied to numeration systems canonically associated with a Pisot number \cite{BH}. Here, we show that this closure property does not hold for Parry-automatic sequences.

\begin{theorem}\label{thm:parry_not_closed}
  There exists a Parry numeration system $U$ such that the class of $U$-automatic sequences is not closed under taking
  image by a uniform substitution.
\end{theorem}

  Throughout this section, we shall consider a specific numeration system $U$ given by the recurrence
  \begin{equation}\label{eq:Un}
    U_n = 3U_{n-1} + 2U_{n-2} + 3U_{n-4}
  \end{equation}
  with initial values $U_0 = 1$, $U_1 = 4$, $U_2 = 15$, and $U_3 = 54$ (it is from \cite[Example~3]{Frougny1997}). The
  characteristic polynomial has two real roots $\beta$ and $\gamma$ and two complex roots with modulus less than $1$.
  We have $\beta \approx 3.61645$ and $\gamma \approx -1.09685$. Thus from the basic theory of linear recurrences, we
  obtain that $U_n \sim c\beta^n$ for some constant $c$. The characteristic polynomial of the recurrence is the minimal
  polynomial of $\beta$ so, in particular, $\gamma$ is a conjugate of $\beta$. Since $\abs{\gamma} > 1$, the number
  $\beta$ is neither a Pisot number nor a Salem number. It is however a Parry number, as it is readily checked that
  $d_\beta(1) = 3203$. Thus $U$ is a Parry numeration system. Moreover, we have $U = U_\beta$. Recall that $\rep_U(\N)$
  is regular as this holds for all Parry numeration systems.

  Consider the characteristic sequence $\infw{x}$ of the set $\{U_n\mid n\ge 0\}$:
  \begin{equation*}
    \infw{x} = 0100100000000001000000000000 \cdots.
  \end{equation*}
From \autoref{the:Uautomatic}, this sequence is $U$-automatic. We consider the constant length substitution $\mu\colon 0\mapsto 0^t, 1\mapsto 10^{t-1}$ with $t \geq 4$. Observe that $\mu(\infw{x})$ is the characteristic sequence of $\{tU_n\mid n\ge 0\}$. The multiplier $4$ is the first interesting value to consider because for $j=2,3$ we have $\rep_U(\{jU_n\mid n\ge 0\})=j0^*$, and we trivially get $U$-recognizable sets. Our aim is to show that $\mu(\infw{x})$ is not $U$-automatic (see \autoref{cor:notU}). This proves \autoref{thm:parry_not_closed}.

  We begin with an auxiliary result that is of independent interest.

  \begin{proposition}\label{pro:aper}
    Let $r \geq 2$ an integer. If $t$ is an integer such that $4 \leq t \leq \floor{\beta^r}$, then the
    $\beta$-expansion of the number $t/\beta^r$ is aperiodic.
  \end{proposition}

  For the proof, we need the following technical lemma, which is obtained by adapting \cite[Lemma~2.2]{Schmidt1980} to our situation. Since $\beta$ is an algebraic number of degree $4$, it is well-known that every element in $\Q(\beta)$ can be expressed as a polynomial in $\beta$ of degree at most $3$ with coefficients in $\Q$.

  \begin{lemma}\label{lem:schmidt}
    Let $x \in [0, 1) \cap \Q(\beta)$, and write
    \begin{equation*}
      x = q^{-1} \sum_{i = 0}^3 p_i \beta^i
    \end{equation*}
    for integers $q$ and $p_i$. If $d_\beta(x)$ is ultimately periodic, then
    \begin{equation*}
      q^{-1} \sum_{i = 0}^3 p_i \gamma^i = \sum_{i = 1}^\infty d_\beta(x)[i] \gamma^{-i}.
    \end{equation*}
  \end{lemma}

  \begin{proof}[Proof of \autoref{pro:aper}]
    Let us first make an additional assumption that $t \geq \ceil{\beta^{r-1}}$ and prove the result in this case. Set
    $x = t/\beta^r =q^{-1} \sum_{i = 0}^3 p_i \beta^i$, and assume for a contradiction that $d_\beta(x)$ is ultimately periodic. Write
    $d_\beta(x) = d_1 d_2 \cdots$. Since $\beta$ and $\gamma$ are conjugates, 
$$\frac{t}{\gamma^r} = q^{-1} \sum_{i = 0}^3 p_i \gamma^i$$
and it follows from \autoref{lem:schmidt}
    that
    \begin{equation*}
      \frac{t}{\gamma^r} = \sum_{i = 1}^\infty d_i \gamma^{-i}.
    \end{equation*}
    In other words, for any positive integer $k$, we have
    \begin{equation}\label{eq:p1}
      t = \sum_{i = 1}^\infty d_i \gamma^{-i+r} = S_{1,k} + S_{k+1,\infty},
    \end{equation}
    where $S_{m,n} := \sum_{i = m}^n d_i \gamma^{-i+r}$. Since $\gamma$ is negative and $d_i \leq 3$ for all $i \geq 1$,
    we have
    \begin{equation}\label{eq:p2}
      S_{r+1,\infty} = \sum_{i = 1}^\infty d_{i+r} \gamma^{-i} \leq 3 \sum_{i = 1}^\infty \gamma^{-2i} = \frac{3\gamma^{-2}}{1-\gamma^{-2}} < 15.
    \end{equation}
    Similarly by discarding the odd terms and estimating $d_i \leq 3$, we
    obtain
    \begin{equation}\label{eq:p3}
      S_{1,r} = \sum_{i = 0}^{r-1} d_{r-i} \gamma^i \leq \frac{3(1 - \gamma^{2(k+1)})}{1 - \gamma^2},
    \end{equation}
    where $k$ is the largest integer such that $2k \leq r - 1$. Combining \eqref{eq:p1}, \eqref{eq:p2}, and
    \eqref{eq:p3} with our assumption $t \geq \ceil{\beta^{r-1}}$, we obtain that
    \begin{equation}\label{eq:p4}
      \beta^{r-1} < \frac{3(1 - \gamma^{2(k+1)})}{1 - \gamma^2} + 15
    \end{equation}
    The left side of \eqref{eq:p4} clearly increases faster than the right side when $r \to \infty$ since
    $\beta \approx 3.62$ and $\gamma^2 \approx 1.20$. Using these approximations, it is straightforward to compute that
    for $r = 4$ the left side of \eqref{eq:p4} is approximately $47$ while the right side is only approximately $22$.
    Hence it must be that $r \leq 3$.

    We are thus left with a few cases we have to deal with separately. The idea is the same, but we need to
    actually compute some digits $d_i$. Suppose first that $r = 3$. Like previously, we see that
    $S_{4,\infty} \leq 3\gamma^{-4}/(1-\gamma^{-2}) < 12.28$. Since $14=\lceil\beta^2\rceil\le t\le \lfloor\beta^3\rfloor=47$, by enumerating all possibilities for the word
    $d_1 d_2 d_3$ (within the given range for $t$), we see that $f(t)=t - S_{1,3}$ is minimized when
    $t = \ceil{\beta^2} = 14$. 

    \begin{figure}[h!tb]
        \centering
        \begin{minipage}{.28\linewidth}
            $$
            \begin{array}{c|c|c}
              t& d_\beta(t/\beta^3) & t-S_{1,3} \\
\hline
14 & 100\cdots & 12.797 \\
 15 & 101\cdots & 12.797 \\
 16 & 102\cdots & 12.797 \\
 17 & 110\cdots & 16.894 \\
 18 & 111\cdots & 16.894 \\
\vdots & \vdots & \vdots \\
 44 & 311\cdots & 40.488 \\
 45 & 312\cdots & 40.488 \\
 46 & 313\cdots & 40.488 \\
 47 & 320\cdots & 45.584 \\
            \end{array}$$
        \end{minipage}
        \begin{minipage}{.58\linewidth}
            \begin{center}
                \scalebox{.8}{\includegraphics{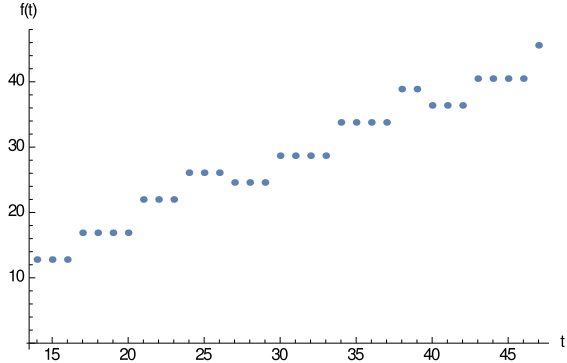}}
            \end{center}
        \end{minipage}
        \caption{Values of $t-S_{1,3}$}
        \label{fig:S13}
    \end{figure}

In this case, $d_1 d_2 d_3 = 100$ and $t - S_{1,3} > 12.79$. This contradicts
    \eqref{eq:p1}. Suppose then that $r = 2$. We proceed as above, but now we are interested in the number
    $t - S_{1,12}$ instead. By enumerating all possibilities, we see that $t - S_{1,12}$ is minimized for
    $t = \ceil{\beta} = 4$. Then $d_1 \cdots d_{12} = 101111202300$ and $t - S_{1,12} > 5.38$. Since
    $S_{13,\infty} < 5$, we get a contradiction.

    Suppose finally that $4 \leq t < \ceil{\beta^{r-1}}$. For $r = 2$, we have already proven the claim because
    $\ceil{\beta} = 4$, so we may suppose that $r > 2$. As $t < \beta^{r-1}$, we see that $\beta \cdot t/\beta^r < 1$
    meaning that $d_1 = 0$. Thus $d_2 d_3 \cdots$ is the $\beta$-expansion of $t/\beta^{r-1}$. Inductively it follows
    that this expansion is aperiodic, and we are done.
  \end{proof}

  We now have tools to prove that the set $\{t U_n \mid n \geq 0\}$ is not $U$-recognizable for $t \geq 4$.

  \begin{corollary}\label{cor:notU}
    The set $\{t U_n\mid n\ge 0\}$ is not $U$-recognizable for $t \geq 4$. In other words, its characteristic sequence
    $\mu(\infw{x})$ is not $U$-automatic.
  \end{corollary}
  \begin{proof}
    Let $t \geq 4$, and suppose that $\ceil{\beta^{r-1}} \leq t \leq \floor{\beta^r}$ for some $r \geq 2$. Recall that
    $U_n \sim c\beta^n$ for some positive constant $c$. By some simple asymptotic analysis, it is easy to
    see that $U_{n+r-1} < tU_n < U_{n+r}$ for $n$ large enough. Hence, for $n$ large enough, $\rep_U(tU_n)$ is a word
    of length $n+r$ (starting with a nonzero digit). Let $k > 0$. We show that, for large enough $n$, $\rep_U(tU_n)$
    and $d_\beta(t/\beta^r)$ have the same prefix of length $k$. See \autoref{tab:4un} for an example. Assume that
    \begin{equation*}
      \rep_U(tU_n)=d_1\cdots d_k d_{k+1}\cdots d_{n+r}.
    \end{equation*}
    The extremal values for $d_{k+1} \cdots d_{n+r}$ are $0^{n+r-k}$ and (possibly) $\rep_U(U_{n+r-k}-1)$ due to the greediness of the representations. Hence
    \begin{equation*}
      0 \leq tU_n - d_1U_{n+r-1} - \cdots - d_kU_{n+r-k} < U_{n+r-k}.
    \end{equation*}
    Dividing by $U_{n+r}$ and by letting $n$ tend to infinity, we get
    \begin{equation*}
      0 \leq \frac{t}{\beta^r} - \frac{d_1}{\beta} - \cdots - \frac{d_k}{\beta^k} <\frac{1}{\beta^k}.
    \end{equation*}
    Otherwise stated, the first $k$ digits of $d_\beta(t/\beta^r)$ are $d_1 \cdots d_k$.
  
    Now proceed by contradiction and assume that $\rep_U(\{tU_n \mid n \geq 0\})$ is accepted by a finite deterministic
    automaton. By a classical pumping argument, there exist $x, y, z \in A_U^*$, with $y$ nonempty, such that $xy^jz$
    is accepted by this automaton for all $j \geq 0$. Hence, $d_\beta(t/\beta^r)$ should be of the form $xy^\omega$
    contradicting \autoref{pro:aper}.
  \end{proof}

% Michel "I like this illustration" ;)

  \begin{table}
      $$
\begin{array}{c|ccccccccccc}
1 &1 & 0 &  &  &  &  &  &  &  &  &  \\
2 &1 & 0 & 1 &  &  &  &  &  &  &  &  \\
3 &1 & 0 & 1 & 1 &  &  &  &  &  &  &  \\
4 &1 & 0 & 1 & 1 & 1 &  &  &  &  &  &  \\
5 &1 & 0 & 1 & 1 & 1 & 1 &  &  &  &  &  \\
6 &1 & 0 & 1 & 1 & 1 & 1 & 2 &  &  &  &  \\
7 &1 & 0 & 1 & 1 & 1 & 1 & 2 & 0 &  &  &  \\
8 &1 & 0 & 1 & 1 & 1 & 1 & 2 & 0 & 3 &  &  \\
9 &1 & 0 & 1 & 1 & 1 & 1 & 2 & 0 & 2 & 3 &  \\
10 &1 & 0 & 1 & 1 & 1 & 1 & 2 & 0 & 2 & 3 & 0 \\
\end{array}$$
\caption{Representations of the first $4U_n$.}
      \label{tab:4un}
  \end{table}

  \autoref{cor:notU} is interesting because it shows that addition in $U$ is not computable by a finite automaton, i.e., its graph is not a regular language. Indeed, if this was the case, then surely multiplication by any constant would be computable by a finite automaton contrary to \autoref{cor:notU}. This result is not new; it already appears in \cite[Example~3]{Frougny1997}. The conclusion is that addition in a Parry numeration system is not necessarily computable by a finite automaton. This shows in particular that Parry-recognizable sets do not have a characterization based on first-order logic like Pisot-recognizable sets have. This is a considerable defect of Parry numeration systems that are not Pisot.

  Let us then describe why a word obtained from a $k$-automatic sequence by periodically deleting letters is still
  $k$-automatic. Suppose that $\infw{x}$ is a $k$-automatic sequence over $A$, and let $\infw{y}$ be the word obtained
  from $\infw{x}$ by keeping only the letters at positions $0$, $t$, $2t$, $3t$, $\ldots$ for a fixed integer $t \geq
  2$. In other words, we have $\infw{y}[n] = \infw{x}[tn]$. As mentioned at the beginning of this section, for each
  $a \in A$, there exists a first-order formula $\varphi_a(n)$ in $\langle \N, +, V_k \rangle$ such that it holds if
  and only if $\infw{x}[n] = a$. By substituting $n$ by $tn$ in $\varphi_a(n)$, we obtain a new first-order formula in
  $\langle \N, +, V_k \rangle$ such that it holds if and only if $\infw{y}[n] = a$. It follows from
  \autoref{the:Uautomatic} that $\infw{y}$ is $k$-automatic. Again, similar construction works in the Pisot case.

  Let us next show that the class of $U$-automatic sequences is not closed under periodic deletion. Consider the
  characteristic sequence $\infw{y}$ of the set $\{U_n/2 \mid \text{$n \geq 0$ and $U_n$ is even}\}$:
  \begin{equation*}
    \infw{y} = 0010000000000000000000000001 \cdots.
  \end{equation*}
  This sequence $\infw{y}$ is obtained from the characteristic sequence $\infw{x}$ of the set $\{U_n \mid n \geq 0\}$ by removing
  its every second letter. Indeed, $\infw{y}[n]=\infw{x}[2n]$ hence $\infw{y}[n]=1$ if and only if $2n$ belongs to $\{U_j \mid j \geq 0\}$. We will show that $\infw{y}$ is not $U$-automatic, which proves the following theorem.

  \begin{theorem}
    There exists a Parry numeration system $U$ such that the class of $U$-automatic sequences is not closed under
    periodic deletion.
  \end{theorem}
  
  Let us begin with the following result.

  \begin{proposition}\label{pro:aper12}
    The $\beta$-expansion of $1/2$ is aperiodic.
  \end{proposition}
  \begin{proof}
    Assume for a contradiction that $d_\beta(1/2) = d_1 d_2 \cdots$ is ultimately periodic. Like in the proof of
    \autoref{pro:aper}, we obtain that
    \begin{equation*}
      \frac12 = \sum_{i = 1}^\infty d_i \gamma^{-i} = S_{1,k} + S_{k+1,\infty},
    \end{equation*}
    where $S_{r,s} = \sum_{i = r}^s d_i \gamma^{-i}$. It can be computed that
    $d_1 \cdots d_{21} = 123102303001010220123$. This computation actually needs some extra accuracy. It is sufficient
    to know that $3.61645454325$ are correct initial digits for $\beta$. Using this information on $d_1 \cdots d_{21}$,
    it is computed that
    \begin{equation*}
      S_{1,21} < -2.20.
    \end{equation*}
    Since $\gamma$ is negative and $d_i \leq 3$ for all $i \geq 1$, we obtain that
    \begin{equation*}
      S_{22,\infty} \leq \frac{3\gamma^{-22}}{1-\gamma^{-2}} < 2.33.
    \end{equation*}
    The two preceding inequalities show that $1/2 < -2.20 + 2.33 = 0.13$, which is obviously absurd.
  \end{proof}

  Interestingly the $\beta$-expansion of $1/3$ is ultimately periodic. Indeed, it can be shown that
  $d_\beta(1/3) = 10(2212)^\omega$.

  \begin{corollary}
    The set $\{U_n/2 \mid \text{$n \geq 0$ and $U_n$ is even}\}$ is not $U$-recognizable. In other words, its
    characteristic sequence $\infw{y}$ is not $U$-automatic.
  \end{corollary}
  \begin{proof}
    We follow steps similar to those of the proof of \autoref{cor:notU}. From \eqref{eq:Un}, it is clear that
    $U_{n-1} < \floor{U_n/2} < U_n$ for $n > 1$, so that $\rep_U(\floor{U_n/2})$ is a word of length $n$. Let $k > 0$.
    We show that, for large enough $n$, $\rep_U(\floor{U_n/2})$ and $d_\beta(1/2)$ have the same prefix of length $k$.
    Assume that
    \begin{equation*}
      \rep_U(\floor{U_n/2}) = d_1 \cdots d_k d_{k+1} \cdots d_n.
    \end{equation*}
    Again, the extremal values for $d_{k+1} \cdots d_n$ are $0^{n-k}$ and (possibly) $\rep_U(U_{n-k}-1)$ due to the greediness of the representations. Therefore
    \begin{equation*}
      0 \leq \floor{U_n/2} - d_1 U_{n-1} - \cdots - d_k U_{n-k} < U_{n-k}.
    \end{equation*}
    Clearly $\floor{U_n/2}/U_n \xrightarrow{n \to \infty} 1/2$ so, dividing by $U_n$ and by letting $n$ tend to
    infinity, we obtain
    \begin{equation*}
      0 \leq \frac12 - \frac{d_1}{\beta} - \cdots - \frac{d_k}{\beta^k}  < \frac{1}{\beta^k}.
    \end{equation*}
    Thus the first $k$ digits of $d_\beta(\floor{U_n/2})$ are $d_1 \cdots d_k$. This means that the words of the
    language $\rep_U(\{U_n/2 \mid \text{$n \geq 0$ and $U_n$ is even}\})$ share longer and longer prefixes with
    $d_\beta(1/2)$.

    The results follows by an argument similar to the final paragraph of the proof of \autoref{cor:notU}: if
    $\rep_U(\{U_n/2 \mid \text{$n \geq 0$ and $U_n$ is even}\})$ is accepted by a finite deterministic automaton, then
    $d_\beta(1/2)$ is ultimately periodic, and this is impossible by \autoref{pro:aper12}.
  \end{proof}

  Notice that the proof in fact shows that the set $\{\floor{U_n/2} \mid n \geq 0\}$ is not $U$-recognizable. Even
  though $\infw{y}$ is not $U$-automatic, we suspect that the word obtained from $\infw{x}$, the characteristic
  sequence of $\{U_n \mid n \geq 0\}$, by keeping only the letters at indices that are divisible by $3$ is
  $U$-automatic. This would follow from our conjecture that $\rep_U(\{U_n/3 \mid \text{$n \geq 0$ and
  $U_n \equiv 0 \pmod{3}$}\})$ equals $11 + 10(2212)^*(3 + 23 + 222 + 2213)$, but we have not attempted to rigorously
  prove this. Notice that $U_n$ is divisible by $3$ when $n \geq 2$.

\section{Multidimensional sequences}\label{sec:multi}
By \autoref{prp:kernel_1D}, an infinite word is $U$-automatic with respect to a numeration system $U$ with $\rep_U(\N)$
regular if and only its $U$-kernel is finite. Moreover, this is true more generally for abstract numeration systems.
The generalization of this result to multidimensional sequences $\infw{x} = (x_{m,n})_{m,n \geq 0}$
\cite[Prop.~32]{RigoMaes} is however slightly problematic as an extra assumption on the projections
$(x_{k,n})_{n \geq 0}$ and $(x_{m,k})_{m \geq 0}$ is required. This extra assumption is however unnecessary for
positional numeration systems considered in this paper. This fact is not found in the literature according to our
knowledge, and this section is devoted to filling this gap.

For the sake of simplicity of presentation, we limit our presentation to two-dimensional sequences. We will consider
finite automata reading pairs of digits. In particular, a pair of words can be read only if the two components have the
same length. With positional numeration systems, when considering two representations of different length, then the
shorter is padded with leading zeroes. For general abstract numeration systems an additional padding letter needs to be
added, and this causes some complications.

\begin{definition}
  Let $U$ be a numeration system. A bi-infinite word $\infw{x}=(x_{m,n})_{m,n\ge 0}$ over an alphabet $B$ is \emph{$U$-automatic} if there exists a complete DFAO $(Q,q_0,A_U\times A_U,\delta,\tau)$ with transition function $\delta\colon Q\times (A_U\times A_U)^*\to Q$ and output function $\tau\colon Q\to B$ such that $\delta(q_0,(0,0))=q_0$ and
  \begin{equation*}
    x_{m,n}=\tau(\delta(q_0,(0^{\ell-|\rep_U(m)|}\rep_U(m),0^{\ell-|\rep_U(n)|}\rep_U(n)))),\quad \forall m,n\ge 0,
  \end{equation*}
where $\ell=\max\{|\rep_U(m)|,|\rep_U(n)|\}$.
The bi-infinite word $\infw{x}$ is \emph{$k$-automatic} (resp. \emph{Parry-automatic}, \emph{Bertrand-automatic}) if
$U=(k^n)_{n\ge 0}$ for an integer $k \geq 2$ (resp. $U$ is a Parry numeration system, $U$ is a Bertrand numeration system).
\end{definition}

\autoref{def:Ukernel} is extended as follows (we make use of the notation $i(s,n)$ introduced therein).

\begin{definition}
    The \emph{$U$-kernel} of a bi-infinite word $\infw{x}=(x_{m,n})_{m,n\ge 0}$ over $B$ is the set 
\begin{equation*}
  \ker_U(\infw{x}):=\{(x_{i(s,m),i(t,n)})_{m,n\ge 0}\in B^{\N^2} \mid s,t\in A_U^*,\ |s|=|t|\}.
\end{equation*}
\end{definition}

Let us then state and prove the result mentioned above.

\begin{proposition}
    Let $U$ be a numeration system such that the numeration language $\rep_U(\N)$ is regular. A bi-infinite word $\infw{x}=(x_{m,n})_{m,n\ge 0}$ is $U$-automatic if and only if its $U$-kernel is finite. 
\end{proposition}

\begin{proof}
  Let $\infw{x} = (x_{m,n})_{m,n \geq 0}$ be a bi-infinite word. From \cite[Prop.~32]{RigoMaes}, we already know that
  if $\infw{x}$ is $U$-automatic, then its $U$-kernel is finite because the result holds for all abstract numeration
  systems. We only need to prove the converse. 

  Denote the $U$-kernel of $\infw{x}$ by $K$, and suppose that it is finite. For $s\in A_U^*$, define
  \begin{equation*}
    \mathcal{L}(s):=0^*\rep_U(\N)\cdot s^{-1}=\{w\in A_U^*\mid ws \in 0^*\rep_U(\N)\}.
  \end{equation*}
  By assumption $0^*\rep_U(\N)$ is regular, so it follows from the Myhill--Nerode theorem that the set of right
  quotients 
  \begin{equation*}
    J:=\{\mathcal{L}(s)\mid s\in A_U^*\}
  \end{equation*}
  is finite. Let us define a DFAO $\mathcal{M}$ with state set
  \begin{equation*}
    Q:=J\times J\times K,
  \end{equation*}
  transition function $\delta$, output function $\tau$, and initial state
  \begin{equation*}
q_0:=
\left(
0^*\rep_U(\N),0^*\rep_U(\N),(x_{m,n})_{m,n\ge 0}
\right)
=\left(
\mathcal{L}(\varepsilon),\mathcal{L}(\varepsilon),(x_{i(\varepsilon,m),i(\varepsilon,n)})_{m,n\ge 0}
\right).
  \end{equation*}
  For a state $q=(\mathcal{L}(s),\mathcal{L}(t),(x_{i(s,m),i(t,n)})_{m,n\ge 0})$ in $Q$, with $|s|=|t|$, and each pair $(a,b)$ of digits
  in $A_U \times A_U$, we set
  \begin{equation*}
    \delta(q,(a,b))=(\mathcal{L}(as),\mathcal{L}(bt),(x_{i(as,m),i(bt,n)})_{m,n\ge 0}).
  \end{equation*}
  For other types of states, i.e., $(\mathcal{L}(s),\mathcal{L}(t),(x_{i(s',m),i(t',n)})_{m,n\ge 0})$ with $s\neq s'$ or $t\neq t'$, we leave the transition function undefined as it is clear that such states are not
  reachable from the initial state $q_0$.

  We have to check that the transition function $\delta$ is well-defined. Assume that 
  \begin{equation*}
    (\mathcal{L}(s),\mathcal{L}(t),(x_{i(s,m),i(t,n)})_{m,n\ge 0})
    =(\mathcal{L}(s'),\mathcal{L}(t'),(x_{i(s',m),i(t',n)})_{m,n\ge 0})
  \end{equation*}
  with $|s|=|t|$ and $|s'|=|t'|$. For all $(a,b)\in A_U\times A_U$, we need to show that 
  \begin{equation*}
    (\mathcal{L}(as),\mathcal{L}(bt),(x_{i(as,m),i(bt,n)})_{m,n\ge 0})
    =(\mathcal{L}(as'),\mathcal{L}(bt'),(x_{i(as',m),i(bt',n)})_{m,n\ge 0}).
  \end{equation*}
  For the first two components, the result follows from the definition: $\mathcal{L}(as)=\mathcal{L}(s)\cdot a^{-1}$
  for any letter $a$. For the third component, we want to prove that $x_{i(as,m),i(bt,n)}=x_{i(as',m),i(bt',n)}$ for
  all $m,n\ge 0$. We know that $\mathcal{L}(s)=\mathcal{L}(s')$, $\mathcal{L}(t)=\mathcal{L}(t')$, and
  $x_{i(s,m),i(t,n)}=x_{i(s',m),i(t',n)}$ for all $m,n\ge 0$. Enumerate the words of $\mathcal{L}(s)\setminus 0A_U^*$
  in genealogical order $\prec$:
  \begin{equation*}
    \mathcal{L}(s)\setminus 0A_U^*=\{ r_{s,0} \prec r_{s,1} \prec r_{s,2} \prec \cdots\}.
  \end{equation*}
  Similarly, we write
  \begin{equation*}
    \mathcal{L}(t)\setminus 0A_U^*=\{ r_{t,0} \prec r_{t,1} \prec r_{t,2} \prec \cdots\}.
  \end{equation*}
  Note that if $s$ is a suffix occurring in a valid $U$-representation, then $r_{s,0}=\varepsilon$; similarly for
  $r_{t,0}$. Let $j,k\ge 0$. Since $r_{s,j}$ and $r_{s,k}$ do not start with a zero digit, we have
  \begin{equation*}
    r_{s,j}\prec r_{s,k} \, \Leftrightarrow \, \val_U(r_{s,j} 0^{|s|})<\val_U(r_{s,k} 0^{|s|}),
  \end{equation*}
  and an analogous equivalence holds for $r_{t,j}$ and $r_{t,k}$. The subsequence $(x_{i(s,m),i(t,n)})_{m,n\ge 0}$ is
  the same as the sequence
  \begin{equation*}
    (x_{\val_U(r_{s,m}s),\val_U(r_{t,n}t)})_{m,n\ge 0}
  \end{equation*}
  because by definition~\eqref{eq:Is}, $i(s,m)$ (resp. $i(t,n)$) is the $m$th (resp. $n$th) integer belonging to
  $\mathcal{I}_s=\val_U(0^*\rep_U(\N)\cap A_U^*s)$ (resp. $\mathcal{I}_t$). Notice that words in $\mathcal{L}(s)$
  (resp. $\mathcal{L}(t)$) starting with $0$ do not provide any new indices. So when building the subsequence, we can
  limit ourselves to words not starting with $0$. If we select in $\mathcal{L}(s)\setminus 0A_U^*$ all words ending
  with $a$, we get exactly $(\mathcal{L}(as)\setminus 0A_U^*)a$, which is equal to
  $(\mathcal{L}(as')\setminus 0A_U^*)a$ because $\mathcal{L}(as)=\mathcal{L}(as')$. Let $m\ge 0$ and $r_{as,m}$ be the
  $m$th word in $\mathcal{L}(as)\setminus 0A_U^*$. Suppose that the $m$th word in $(\mathcal{L}(as)\setminus 0A_U^*)a$,
  which is $r_{as,m}a$, occurs as the $k$th word $r_{s,k}$ in $\mathcal{L}(s)\setminus 0A_U^*$. Then $r_{s,k}$ also
  occurs as the $k$th word $r_{s',k}$ in $\mathcal{L}(s')\setminus 0A_U^*$. With our notation, we have
  \begin{equation*}
    r_{as,m}a=r_{s,k},\ 
    \val_U(r_{as,m}as)=\val_U(r_{s,k}s), \text{ and }
    i(as,m)=i(s,k)=i(s',k).
  \end{equation*}
  We can make similar observations for the other component. Supposing that $r_{bt,n} = r_{t,\ell}$ for some $\ell$, we
  thus have
  \begin{equation*}
    x_{i(as,m),i(bt,n)}=x_{i(s,k),i(t,\ell)} =
    x_{i(s',k),i(t',\ell)}=x_{i(as',m),i(bt',n)},
  \end{equation*}
  where the central equality comes from our initial assumption. Therefore we have shown that $\delta$ is well-defined.

  From our definition of the transition function $\delta$, the accessible part of $\mathcal{M}$ is limited to states
  $q$ of the form
  \begin{equation*}
    (\mathcal{L}(s),\mathcal{L}(t),(x_{i(s,m),i(t,n)})_{m,n\ge 0})
  \end{equation*}
  with $|s|=|t|$. For such a state $q$, we set
  \begin{equation*}
    \tau(q)=x_{i(s,0),i(t,0)}.
  \end{equation*}
  Notice that the preceding arguments show that $\tau$ is also well-defined. To conclude with the proof, let us show
  that if $s,t$ are two words of the same length in $0^*\rep_U(\N)$, then 
  \begin{equation*}
    \tau(\delta(q_0,(s^R,t^R)))=x_{\val_U(s),\val_U(t)},
  \end{equation*}
  where $s^R$ and $t^R$ respectively denote the reversals of the words $s$ and $t$. Reading $(s^R,t^R)$ from $q_0$
  leads to the state $(\mathcal{L}(s),\mathcal{L}(t),(x_{i(s,m),i(t,n)})_{m,n\ge 0})$. Since $s,t\in
   0^*\rep_U(\N)$, we have that $\varepsilon$ belongs to $\mathcal{L}(s)$ and $\mathcal{L}(t)$.
 It is clear that
  ${i(s,0)}={\val_U(s)}$ and ${i(t,0)}={\val_U(t)}$. 

  We have thus proved that $\infw{x}$ is reversal-$U$-automatic. It follows from \autoref{pro:reversal} (which also
  holds in the multidimensional setting) that $\infw{x}$ is $U$-automatic.
\end{proof}

  \section*{Acknowledgments}
  This work was initiated when the second author visited his coauthors at the Department of Mathematics of the
  University of Liège during June 2018. He thanks the department for their hospitality.

\end{document}